\documentclass[letterpaper]{article}
\usepackage[cmex10]{amsmath}
\usepackage{amsthm}
\usepackage{enumerate}
\usepackage{graphicx}
\usepackage{epstopdf}

\usepackage{amsfonts,amscd}
\usepackage{float}

\newtheorem{theorem}{Theorem}[section]
\newtheorem{lemma}[theorem]{Lemma}

\newcommand{\be}{\begin{equation}}
\newcommand{\ee}{\end{equation}}
\newcommand{\bea}{\begin{eqnarray}}
\newcommand{\eea}{\end{eqnarray}}
\newcommand{\beann}{\begin{eqnarray*}}
\newcommand{\eeann}{\end{eqnarray*}}
\newcommand{\ip}[2]{\langle{#1}\mid{#2}\rangle}

\newcommand{\unity}{{1\hskip -3pt \rm{I}}}
\newcommand{\Tr}{{\text{Tr}}}
\begin{document}
\title{Shannon and R\'enyi entropy rates of stationary vector valued Gaussian random processes}
\author{
    Jaideep Mulherkar\\
    Dhirubhai Ambani Institute of \\
    Information and Communication Technology\\
    jaideep\_mulherkar@daiict.ac.in
}
\maketitle
\begin{abstract}
We derive expressions for the Shannon and R\'enyi entropy rates of stationary vector valued Gaussian random processes using the block matrix version of Szeg\"o's theorem.
\end{abstract}
\section{Introduction}

The differential Shannon entropy $H(X)$ of a continuous random variables $X$ with probability density $f_X(x)$  is defined as
\bea
H(X) = -\int f_X(x)\log f_X(x) dx
\eea
Differential entropy of a continuous random variable was introduced in Shannon's original paper \cite{CS1948}. The Shannon Entropy appears in his source coding theorem as a bound the lossless compression possible. The connection between thermodynamics and  information theory was first made by Boltzmann and expressed by his famous equation
\beann
{\displaystyle S=k_{\text{B}}\log(W)} {\displaystyle S=k_{\text{B}}\log(W)}
\eeann
where $S$  is the thermodynamic entropy of a particular macrostate, W is the number of microstates that can yield the given macrostate, and $k_B$ is Boltzmann's constant.
R\'enyi \cite{AR1961} and Tsallis \cite{T1988} generalized the notion of Shannon entropy.
The R\'enyi entropy of order $\alpha$, where $\alpha \geq 0$ and $\alpha \neq 1$, is defined as 
\bea
H_{\alpha}(X) = \frac{1}{1-\alpha} \log\int f_X^{\alpha}(x)dx
\eea
As $\alpha \rightarrow 0$ , the R\'enyi entropy weighs all possible events more equally and the R\'enyi entropy is just the logarithm of the size of the support of $X$. The limit for $\alpha \rightarrow 1$ is the Shannon entropy. As $\alpha$ approaches infinity, the R\'enyi entropy is increasingly determined by the events of highest probability and approaches what is called the min-entropy. The R\'enyi entropy has been widely used in information theory, economics and physics and computer science. In physics has been useful in the analysis of quantum entanglement \cite{BT2002}, uncertainty measures \cite{GL2004, IB2007}, quantum  channel capacity formulas \cite{MH2011} and quantum spin systems \cite{FK2008} and in statistical mechanics \cite{LMD2000}.

Let ${\mathcal{X} = \{X_1,X_2,...,X_n\}}$ be a random process of continuous random variables $X_i$  with joint density function $f_{\mathcal{X}}(x_1,x_2,\ldots,x_n)$ then the joint Shannon and R\'enyi differential entropy is given by
\bea
\label{eq:Ent1}
H(\mathcal{X}) &=& -\int f_{\mathcal{X}}(x_1,x_2,\ldots,x_n) \log f_{\mathcal{X}}(x_1,x_2,\ldots,x_n)d\mathbf{x}\nonumber \\
H_{\alpha}(\mathcal{X}) &=& \frac{1}{1-\alpha}\log\int f_{\mathcal{X}}^{\alpha}(x_1,x_2,\ldots,x_n) d\mathbf{x}
\eea
For a random process $\mathcal{X} = \{X_1,X_2,...,X_n,\ldots\}$ the Shannon and R\'enyi entropy rate is defined as
\bea
\label{eq:Ent2}
\bar{H}(\mathcal{X}) = \lim_{n\rightarrow \infty}\frac{H(\mathcal{X})}{n} \nonumber\\
\bar{H}_{\alpha}(\mathcal{X}) = \lim_{n\rightarrow \infty}\frac{H_{\alpha}(\mathcal{X})}{n}
\eea
Both these limits exist when the process is stationary. Entropy rate is the average amount of information per symbol of the process and is key quantity in information theory. There is a nice formula for the entropy rate of a Markov process \cite{CT1991}. A formula for the R\'enyi entropy rate for Markov processes was obtained in \cite{RFC1999}. For hidden Markov models, a general formula for the entropy rate is still one of the important outstanding problems \cite{MPW2011}. Some results on the entropy rates of special families of hidden Markov processes were obtained in \cite{JSZ2004,HM2010,MMN2012}.  

In this paper we look at the entropy rate of a vector valued Gaussian random process, that is a random process where each co-ordinate is a vector belonging to $\mathbb{R}^m$ for some $m \in \mathbb{N}$ and the random vectors are jointly Gaussian. A scenario where studying vector valued random processes and their entropy rates may be important is in computing the capacities of Multiple Input Multiple Output (MIMO) channels \cite{GJJV2008}. In section \ref{sec:ERMG} we review at some results entropy rates of Gaussian processes. In section \ref{sec:SZTHM} we define the Shannon and R\'enyi entropy rates of vector processes and introduce the block matrix version of Szeg\"o's theorem and in section \ref{sec:ERPRF}  we derive the entropy rate formulas.

\section{Entropy rate of multivariate Gaussian random variables}
\label{sec:ERMG}
If a random variable $X$ is a $N(0, \sigma^2)$ then it is well known \cite{CT1991} that
\bea
H(X) &=& \frac{1}{2} \log 2\pi e \sigma^2\\ \nonumber
H_{\alpha}(X) &=& \frac{1}{2} \log 2\pi \sigma^2 \alpha^{\frac{1}{\alpha-1}}
\eea
For a zero mean multivariate Gaussian distribution process $ \mathcal{X}=\{X_1,X_2,\ldots, X_n\}$ the density function is  $f_{\mathcal{X}}(x_1,x_2,\ldots,x_n) = \frac{1}{\sqrt{(2\pi)^n \det(K_n)}}e^{\frac{1}{2}x^TK_n^{-1}x}$ where
\beann
(K_n)_{ij} = \text{Cov}(X_i,X_j)
\eeann
is the covariance matrix. 
From equations (\ref{eq:Ent1}) and (\ref{eq:Ent2})  one gets that the Shannon and R\'enyi entropy rates of a stationary Gaussian process are given by
\bea
\label{eq:Ent3}
H(\mathcal{X}) &=& \frac{1}{2} \log 2\pi e + \frac{1}{2}\lim_{n\rightarrow\infty} \frac{\log\det(K_n)}{n}\\ \nonumber
H_{\alpha}(\mathcal{X}) &=& \frac{1}{2} \log 2\pi \alpha^{\frac{1}{\alpha-1}} + \frac{1}{2}\lim_{n\rightarrow\infty} \frac{\log\det(K_n)}{n}
\eea
For a stationary Gaussian process  the covariance matrix is a Toeplitz matrix with entries $K(j) := (K_n)_{i,i+j}$.
\beann
\label{eq:Cov}
K_n = \begin{pmatrix}{K}(0) & {K}(-1) & \cdots& {K}(-n+1) \\
	{K}(1) & {K}(0)  & \cdots& {K}(-n+2) \\
	\vdots &   \vdots        & \vdots & \vdots     \\
	{K}(n-1) & {K}(n-1)  & \cdots& {K}(0) 
\end{pmatrix}
\eeann
$K(j)$ is called the autocorrelation function of the Gaussian process and its Fourier transform is the power spectral density $S(\lambda)$: 
\beann
S(\lambda) = \sum_{m= -\infty}^{\infty} K(m) e^{-im\lambda} 
\eeann
As $n\rightarrow \infty$ the density of the eigenvalues of the covariance matrix tends to a limit and the equation (\ref{eq:Ent3}) can be shown to have a particularly nice form that relates to the power spectral density $S(\lambda)$. It was shown by Kolmorogov \cite{K1958} and Golshani et al \cite{GP2010} respectively that the Shannon and R\'enyi entropy rates of Gaussian processes is given by 
\bea
\bar{H}(\mathcal{X}) &=& \frac{1}{2} \log 2\pi e + \frac{1}{4\pi} \int_{-\pi}^{\pi} \log S(\lambda)  d\lambda \\ \nonumber
\bar{H}(\mathcal{X}) &=& \frac{1}{2} \log 2\pi \alpha^{\frac{1}{\alpha-1}} + \frac{1}{4\pi} \int_{-\pi}^{\pi} \log S(\lambda)  d\lambda \\ \nonumber
\eea

\section{Gaussian random vectors and block matrix version of Szeg\"o's theorem}
\label{sec:SZTHM}
Consider a stationary vector valued random process $\mathcal{X} =\{\mathbf{X_1},\mathbf{X_2},...,\mathbf{X_n}\}$, where each coordinate is a random vector taking values in $\mathbb{R}^m$. The random vectors $\mathbf{X_i}$ are jointly Gaussian.  The joint density function $f_{\mathcal{X}}:\mathbb{R}^{nm}\rightarrow \mathbb{R}$ is
\beann
f_{\mathcal{X}}(\mathbf{x_1,x_2,\ldots,x_n}) = \frac{1}{\sqrt{(2\pi)^{nm} (\det \hat{K}_n) } } e^{-\frac{1}{2} (\mathbf{x_1,x_2,\ldots,x_n})^T \hat{K}_{n}^{-1}(\mathbf{x_1,x_2,\ldots,x_n}})
\eeann 
where each $\mathbf{x_i} \in \mathbb{R}^m$ and $\hat{K}_{n}$ is the $nm\times nm$ covariance matrix which is in block Toeplitz form 
\bea
\label{eq:Cov}
\hat{K}_n = \begin{pmatrix}\hat{K}(0) & \hat{K}(-1) & \cdots& \hat{K}(-n+1) \\
                          \hat{K}(1) & \hat{K}(0)  & \cdots& \hat{K}(-n+2) \\
                          \vdots &   \vdots        & \vdots & \vdots     \\
                          \hat{K}(n-1) & \hat{K}(n-1)  & \cdots& \hat{K}(0) 
                          \end{pmatrix}
\eea
Each $\hat{K}(i)$ is a $m\times m$ matrix. We define the joint entropy of the random stationary process of Gaussian random vectors equivalently as
\beann
H(\mathcal{X}) = -\int_{\mathbb{R}^m\times\cdots\times \mathbb{R}^m} f_{\mathcal{X}}(\mathbf{x_1,x_2,\ldots,x_n}) \log f_{\mathcal{X}}(\mathbf{x_1,x_2,\ldots,x_n})\mathbf{dx_1\ldots dx_n}
\eeann
and the entropy rate per coordinate vector as
\bea
\bar{H}(\mathcal{X}) = \lim_{n\rightarrow \infty} \frac{H(\mathcal{X})}{n}
\eea

In computing formulas for the entropy rates, determinants of Toeplitz and block Toeplitz matrices will play an important role. Szeg\"o's  limit theorems describe the asymptotic behavior of the determinants of large Toeplitz matrices \cite{S1915,S1952}. Szeg\"o's  limit theorems have been generalized for block Toeplitz matrices \cite{W1974, W1976}. Toeplitz matrices and determinants have been useful in the study of determinantal processes, integrable models, and entanglement entropy \cite{J2002, BE2007,IMM2008,BSK2006}. We will use the following version of Szeg\"o's limit theorem for block Toeplitz matrices. 
 
\begin{theorem}
\label{thm:Szego}
Let $T:[-\pi,\pi]\rightarrow \mathcal{M}_d$ be continuous matrix valued function with Fourier coefficients by
\beann
\hat{T}(k) = \frac{1}{2\pi}\int_{-\pi}^{\pi} T(\theta) e^{-ik\theta} d\theta \in \mathcal{M}_d
\eeann
then for any absolutely continuous function $f:[\inf T, \sup T]\rightarrow \mathbb{R}$
\beann
\lim_{n\rightarrow \infty} \frac{1}{n} \Tr( f(\hat{T}_n)) = \frac{1}{2\pi} \int_{-\pi}^{\pi}\Tr(f(T(\theta))d\theta
\eeann
\end{theorem}

\section{Shannon and R\'enyi entropy rates of Gaussian random vectors}
\label{sec:ERPRF}
Let $\mathcal{X} =\{\mathbf{X_1},\mathbf{X_2},...,\mathbf{X_n},\ldots\}$ be a vector valued stationary Gaussian process of random vectors with each $\mathbf{X_i}$ is a random vector taking values in $\mathbb{R}^m$. The covariance matrix $\hat{K}_n$ of  $\mathcal{X}$ is a $nm\times nm$ block Toeplitz matrix of the form of equation (\ref{eq:Cov}). Let $K(\theta)$ be the matrix valued Fourier coefficients corresponding to $\hat{K}_n$ as given by theorem (\ref{thm:Szego}). We have the following theorem:

\begin{theorem}
The Shannon entropy rate per coordinate of the vector valued Gaussian random process  $\mathcal{X}$ is given by
\beann
\bar{H}(\mathcal{X}) = \frac{m}{2} \log 2\pi e + \frac{1}{4\pi} \int_{-\pi}^{\pi} \Tr (\log K(\theta)) d\theta
\eeann
\end{theorem}
\begin{proof}
We have
\beann
\bar{H}(\mathcal{X} &=& -\lim_{n\rightarrow \infty} \frac{1}{n} \int f_{\mathcal{X}}(\mathbf{x}) \log f_{\mathcal{X}}(\mathbf{x})d\mathbf{x}\\
&=& -\lim_{n\rightarrow \infty} \frac{1}{n} \int f_{\mathbf{X_1,...,X_n}}(\mathbf{x})\Big[-\frac{1}{2}\mathbf{x}^T \hat{K}_{n}^{-1}\mathbf{x} - \log (2\pi)^{\frac{nm}{2}}(\det \hat{K}_n)^{\frac{1}{2}}\Big]d\mathbf{x} \\
&=& \lim_{n\rightarrow \infty}\Big[ \frac{1}{2n} \int (\mathbf{x}^T \hat{K}_{n}^{-1}\mathbf{x}) f_{\mathbf{X_1,...,X_n}}(\mathbf{x})d\mathbf{x} + \frac{nm}{2n} \log2\pi +  \frac{1}{2n}\log(\det(\hat{K}_n))\Big]
\eeann
Using lemma (\ref{lem:Gauss}) we get
\beann
\int (\mathbf{x}^T \hat{K}_{n}^{-1}\mathbf{x}) f_{\mathbf{X_1,...,X_n}}(\mathbf{x})d\mathbf{x} =  \Tr(\hat{K}_{n}^{-1}\hat{K}_{n}) = \Tr(\unity_{nm}) = nm
\eeann
We also have the identity
\beann
\log (\det(A)) = \Tr (\log (A))
\eeann
Thus we have
\beann
\bar{H}(\mathcal{X}) = \lim_{n\rightarrow \infty} \frac{nm}{2n} + \lim_{n\rightarrow \infty} \frac{nm}{2n} \log 2\pi + \lim_{n\rightarrow \infty} \frac{1}{2n} \Tr(\log(\hat{K}_n))
\eeann
Using theorem (\ref{thm:Szego}) we get
\beann
\bar{H}(\mathcal{X}) = \frac{m}{2} + \frac{m}{2} \log 2\pi + \frac{1}{4\pi} \int_{-\pi}^{\pi} \Tr (\log K(\theta)) d\theta\\
\bar{H}(\mathcal{X}) = \frac{m}{2} \log 2\pi e + \frac{1}{4\pi} \int_{-\pi}^{\pi} \Tr (\log K(\theta)) d\theta
\eeann
\end{proof}
\begin{lemma}
\label{lem:Gauss}
Let $\mathcal{X}$ be a jointly Gaussian random vector zero mean and covariance matrix $K$ and density function $f_{\mathcal{X}}(\mathbf{x})$
\beann
\int (\mathbf{x}^T B\mathbf{x}) f_{\mathcal{X}}(\mathbf{x})d\mathbf{x} = \Tr(BK)
\eeann
\end{lemma}
\begin{proof}
First suppose $B=E_{ij}$ then
\beann
\int (\mathbf{x}^T E_{ij}\mathbf{x}) f_{\mathcal{X}}(\mathbf{x})d\mathbf{x} = \int x_ix_j f_{\mathcal{X}}(\mathbf{x})d\mathbf{x} = K_{ij} = \Tr(E_{ij}K)
\eeann
Now suppose $B= \sum_{ij}b_{ij}E_{ij}$ then we can use the linearity of the inner product $\ip{x}{y} = x^Ty$ and the trace  to get the result.
\end{proof}
\begin{theorem}
	The R\'enyi entropy rate per coordinate of the Gaussian random vector process  $\mathcal{X}$ is given by
	\beann
	\bar{H}_{\alpha}(\mathcal{X}) = \frac{m}{2} \log 2\pi\alpha^{\frac{1}{\alpha-1}} + \frac{1}{4\pi} \int_{-\pi}^{\pi} \Tr (\log K(\theta)) d\theta
	\eeann
\end{theorem}
\begin{proof}	
We have
\beann
\bar{H}_{\alpha}(\mathcal{X}) &=& \lim_{n\rightarrow \infty} \frac{1}{n} \frac{1}{(1-\alpha)}\log\int f_{\mathcal{X}}^{\alpha}(\mathbf{x}) d\mathbf{x}\\
&=& \lim_{n\rightarrow \infty} \frac{1}{n} \frac{1}{(1-\alpha)}\log \int (2\pi)^{-\frac{nm}{2}}(\det \hat{K}_n)^{-\frac{\alpha}{2}}\\ && \qquad\qquad\qquad\qquad \exp{-\frac{1}{2} (\mathbf{x_1,x_2,\ldots,x_n})^T \alpha\hat{K}_{n}^{-1}(\mathbf{x_1,x_2,\ldots,x_n})} d\mathbf{x}\\
&=& \lim_{n\rightarrow \infty} \frac{1}{n} \frac{1}{(1-\alpha)}\log (2\pi)^{-\frac{nm}{2}}(\det \hat{K}_n)^{-\frac{\alpha}{2}} \cdots \\ && \qquad\qquad\qquad \cdots\int\exp{-\frac{1}{2} (\mathbf{x_1,x_2,\ldots,x_n})^T \alpha\hat{K}_{n}^{-1}(\mathbf{x_1,x_2,\ldots,x_n})} d\mathbf{x}
\eeann
Now using the fact about Gaussian integrals that if A is a positive definite matrix then	
\beann
\int \exp\{-\frac{1}{2}x^TAx \}d\mathbf{x}  = \sqrt{\frac{(2\pi)^n}{\det(A)}} 
\eeann
we get
\beann
\bar{H}_{\alpha}(\mathcal{X}) &=& \lim_{n\rightarrow \infty} \frac{1}{n} \frac{1}{(1-\alpha)}\log \Big[ (2\pi)^{-\frac{nm}{2}}(\det \hat{K}_n)^{-\frac{\alpha}{2}} \frac{(2\pi)^{nm/2}}{(\det(\alpha\hat{K}_n^{-1}))^{1/2}}\Big]\\
&=& \lim_{n\rightarrow \infty} \frac{1}{n} \frac{1}{2(1-\alpha)}\log \Big[ (2\pi)^{(1-\alpha)nm}\det (\hat{K}_n)^{1-\alpha} \alpha^{-nm} \Big]\\
&=& \lim_{n\rightarrow \infty} \frac{1}{n} \frac{1}{2(1-\alpha)}\log \Big[ (2\pi \alpha^{\frac{1}{\alpha-1}})^{(1-\alpha)nm}\det (\hat{K}_n)^{1-\alpha}  \Big]\\
&=& \lim_{n\rightarrow \infty} \Big[\frac{nm(1-\alpha)}{2n(1-\alpha)}\log  (2\pi \alpha^{\frac{1}{\alpha-1}}) + \frac{1-\alpha}{2n(1-\alpha)}\log(\det (\hat{K}_n))  \Big]\\
&=& \lim_{n\rightarrow \infty} \Big[\frac{m}{2}\log  (2\pi \alpha^{\frac{1}{\alpha-1}}) + \frac{1}{2n}\log(\det (\hat{K}_n))  \Big]
\eeann
Using the identity
\beann
\log (\det(A)) = \Tr (\log (A))
\eeann
We get
\beann
\bar{H}_{\alpha}(\mathcal{X})&=&  \Big[\frac{m}{2}\log  (2\pi \alpha^{\frac{1}{\alpha-1}}) + \lim_{n\rightarrow \infty} \frac{1}{2n}\Tr\log(\hat{K}_n)  \Big]
\eeann
Finally applying theorem (\ref{thm:Szego}) we get
\beann
	\bar{H}(\mathcal{X}) =\frac{m}{2}\log  (2\pi \alpha^{\frac{1}{\alpha-1}}) + \frac{1}{4\pi} \int_{-\pi}^{\pi} \Tr (\log (K(\theta))) d\theta
\eeann
\end{proof}
\section{Conclusion}
Entropy rates are important in information theory and physics. Exact formulas for entropy rates of stochastic processes are usually not easy to find. Exceptions to this general case are Markov processes, Gaussian processes and a few classes of hidden Markov processes. In this paper we have generalized the formulas for the entropy rates of vector valued stationary Gaussian processes. The main tool for proving our formulas is the block matrix version of the Szeg\"o's theorem. These formulas may have potential applications in channel capacity calculations of classical and quantum information and in the analysis of multi-dimensional data. 
\bibliographystyle{IEEEtran}
\bibliography{GaussianEntropyRate}

\end{document}